\documentclass[11pt,reqno]{amsproc}
\usepackage[margin=1in]{geometry}
\usepackage{amsmath, amsthm, amssymb}
\usepackage{times, esint,stackrel}
\usepackage{enumitem}
\usepackage{xcolor}  
\usepackage[colorlinks=true]{hyperref}
\usepackage[color=yellow]{todonotes}
\usepackage{mathtools}
\usepackage{etoolbox}
\usepackage{mathrsfs}
\usepackage{framed}
\usepackage{csquotes}

\makeatletter
\patchcmd\@thm
  {\let\thm@indent\indent}{\let\thm@indent\noindent}%
  {}{}
\makeatother

\usepackage{color}
\usepackage{hyperref}

\newcommand{\be}{\begin{equation}}
\newcommand{\ee}{\end{equation}}
\newcommand{\bea}{\begin{eqnarray}}
\newcommand{\eea}{\end{eqnarray}}

\newtheorem{thm}{Theorem}
\newtheorem{prop}{Proposition}
\newtheorem{lemma}{Lemma}

\newtheorem{hyp}{Hypothesis}
\newtheorem{quest}{Question}
\theoremstyle{definition}
\newtheorem{rem}{Remark}

\newcommand{\ve}{{\varepsilon}}

\newcommand{\rmd}{{\rm d}}
\newcommand{\bx}{{ {x} }}

\newcommand{\cV}{{ \mathcal{V}}}
\newcommand{\bu}{{{u}}}

\newcommand{\Dlim}{{\mathcal{D}'}\mbox{-}\lim}

\definecolor{ao(english)}{rgb}{0.0, 0.5, 0.0}

\newtheorem{innercustomgeneric}{\customgenericname}
\providecommand{\customgenericname}{}
\newcommand{\newcustomtheorem}[2]{%
  \newenvironment{#1}[1]
  {%
   \renewcommand\customgenericname{#2}%
   \renewcommand\theinnercustomgeneric{##1}%
   \innercustomgeneric
  }
  {\endinnercustomgeneric}
}

\newcustomtheorem{customlemma}{Lemma}
\newcustomtheorem{customhyp}{Hypothesis}
\newcustomtheorem{customques}{Question}

\newcommand{\bq}{\begin{equation}}
\newcommand{\eq}{\end{equation}}
\newcommand{\bqa}{\begin{eqnarray*}}
\newcommand{\eqa}{\end{eqnarray*}}

\newcommand{\T}{\mathbb{T} }

\title[Self--Regularization in turbulence]{Self--Regularization in turbulence from \\  the Kolmogorov 4/5--Law and Alignment }
\author{Theodore D. Drivas}
\address{ Department of Mathematics, Stony Brook University,
Stony Brook, NY, 11794}
\email{tdrivas@math.stonybrook.edu}

\date{today}
\begin{document}

\begin{abstract}
A defining feature of 3D hydrodynamic turbulence is that the rate of energy dissipation is bounded away from zero as viscosity is decreased (Reynolds number  increased).  This phenomenon -- anomalous dissipation --  is sometimes called the `zeroth law of turbulence' as it underpins many celebrated theoretical predictions.  Another robust feature observed in turbulence is that velocity structure functions $S_p(\ell) :=\langle |\delta_\ell u|^p\rangle$ exhibit persistent power-law scaling in the inertial range, namely $S_p(\ell) \sim |\ell|^{\zeta_p}$ for exponents $\zeta_p>0$ over an ever increasing (with Reynolds) range of scales.
This behavior indicates that the velocity field retains some fractional differentiability uniformly in the Reynolds number.
 The Kolmogorov 1941 theory of turbulence predicts that $\zeta_p=p/3$ for all $p$ and Onsager's 1949 theory establishes the requirement that $\zeta_p\leq p/3$ for $p\geq 3$ for consistency with the zeroth law.
Empirically,  $\zeta_2 \gtrapprox 2/3$  and $\zeta_3 \lessapprox 1$,  suggesting that turbulent Navier-Stokes solutions approximate dissipative weak solutions of the Euler equations possessing (nearly) the minimal degree of singularity required to sustain anomalous dissipation. 
In this note, we adopt an experimentally supported hypothesis on the anti-alignment of velocity increments with their separation vectors and demonstrate  that  the inertial dissipation provides a regularization mechanism via the Kolmogorov 4/5--law.  
\end{abstract}

\maketitle

\vspace{-5mm}

\section{Introduction}

We consider spatially periodic, incompressible viscous fluids  governed by the Navier-Stokes equations 
\bea\label{NSE}
\partial_t \bu^\nu + \nabla \cdot (\bu^\nu\otimes \bu^\nu) \!\! &=& \!\! -\nabla p^\nu + \nu \Delta \bu^\nu + f^\nu,\\ \label{incom}
\nabla \cdot u^\nu \!\!&=& \!\!0,
\eea
with solenoidal initial data
$u^\nu|_{t=0}=u_0^\nu\in  L^2({\mathbb T}^d)$ and body forcing $f^\nu \in L^2(0,T;L^2(\mathbb{T}^d))$. The parameter $\nu>0$ is the kinematic viscosity of the fluid. Upon nondimensionalization, it is replaced by the inverse Reynolds number $\mathsf{Re}^{-1} = \nu/\mathsf{ U}\mathsf{ L}$, where $\mathsf{ U}$ is a characteristic velocity and $\mathsf{ L}$ a characteristic length.
  If equations \eqref{NSE} and \eqref{incom} are understood as holding in the sense of distributions on $[0,T]\times \mathbb{T}^d$, then solutions of class $L^\infty(0,T;L^2({\mathbb T}^d))\cap L^2(0,T;H^1({\mathbb T}^d))$  known as Leray solutions \cite{L34}, exist for all time $T>0$ but are not known to be unique.
A fundamental property of these solutions is that they satisfy a global energy inequality.  This means 
that energy dissipation due to the viscosity of the fluid cannot exceed the difference in initial and final kinetic energies plus the energy input by forcing.  This inequality can be restated as an equality  by accounting for the dissipation arising from an inertial cascade to small scales caused by (hypothetical) singularities in the Leray weak solutions \cite{DR00}:
\bea\label{viscousDiss}
\int_0^T\int_{\mathbb{T}^d} \! \varepsilon^\nu[\bu^\nu] \  \rmd \bx \rmd t
= \frac{1}{2}\int_{\mathbb{T}^d} \! | \bu_0^\nu|^2\rmd \bx
-\frac{1}{2}\int_{\mathbb{T}^d} \! |\bu^\nu(\cdot,T)|^2\rmd \bx
+\int_0^T\int_{\mathbb{T}^d} \! \bu^\nu\cdot f^\nu\  \rmd \bx \rmd t,
\eea
for almost every $T\geq 0$, where the total energy dissipation rate is
\be\label{epsDef}
\varepsilon^\nu[\bu^\nu] := \nu |\nabla \bu^\nu|^2+ D[\bu^\nu].
\ee
 The  dissipation due to possible singularities, $D[\bu^\nu]$, is a non-negative distribution (Radon measure).  A consequence of \eqref{viscousDiss} is that the cumulative energy dissipation $\varepsilon[\bu^\nu]$ is bounded by norms of data and forcing.

A striking feature of high-$Re$ turbulence is that energy dissipation does not vanish in the limit of viscosity going to zero. Namely, that there exists a number $\ve>0$ independent of viscosity $\nu$ such that
\be\label{zerothLaw}
\int_0^T\int_{\mathbb{T}^d} \varepsilon^\nu[\bu^\nu] \rmd x \rmd t  \geq {\ve}>0.
\ee
See, e.g. \cite{BO95,TBS02,KRS98,KIYIU03,KRS84,PKW02}.
This phenomenon, known as anomalous dissipation, is so fundamental to our modern understanding of turbulence that it has been termed the ``zeroth law" \cite{F95}.   It should be emphasized however that, to this day,  no single mathematical example of \eqref{zerothLaw} is available, although there has been great progress in understanding similar behavior in some model problems such as 1D conservation laws and compressible flows \cite{KMS00,Dshock21,ED15, DE18},  shell models \cite{CFP09,MSV07,FGV16,AM16b},
and passive scalars \cite{BGK98,LJR02,DEIJ19,BBPS19}.  

Despite its conjectural status from the point of view of mathematics, under the experimentally corroborated assumption that behavior \eqref{zerothLaw} occurs together with some heuristic assumptions on statistical properties (homogeneity, isotropy, monofractal scaling), Kolmogorov \cite{K41} made a remarkable prediction about the structure of turbulent velocity fields at high Reynolds number, namely that
\be\label{SFscaling}
S_p^\|(\ell):= \langle (\delta_\ell u^\nu\cdot \hat{\ell} )^p\rangle  \sim (\ve |\ell|)^{p/3} \qquad \text{for} \qquad \ell_\nu \ll \ell \ll L
\ee
where $\delta_\ell u^\nu(x,t):= u^\nu(x+\ell,t)-u^\nu(x,t)$, $\hat{\ell}= \ell/|\ell|$  and where $\langle\cdot \rangle$ represents some suitable combination of  space, time and ensemble averages. 
The length $\ell_\nu$, known as the Kolmogorov scale, represents a small-scale dissipative cutoff and the integral scale $L$ represents the size of the largest eddy in the flow.  The range $\ell_\nu \ll \ell \ll L$ over which the scaling \eqref{SFscaling} holds is known as the \emph{inertial range}.   The objects $S_p^\|(\ell)$ are called $p$th-order longitudinal structure functions since they measure the (signed) variation of $p$th powers of the velocity increments in the direction of their separation vectors.
See Figure \ref{figure1} for evidence of such persistent inertial-range scaling from numerical simulations of  homogenous isotropic Navier-Stokes turbulence \cite{ISY20}. 

\begin{figure} [h!]
\includegraphics[width=0.47\linewidth]{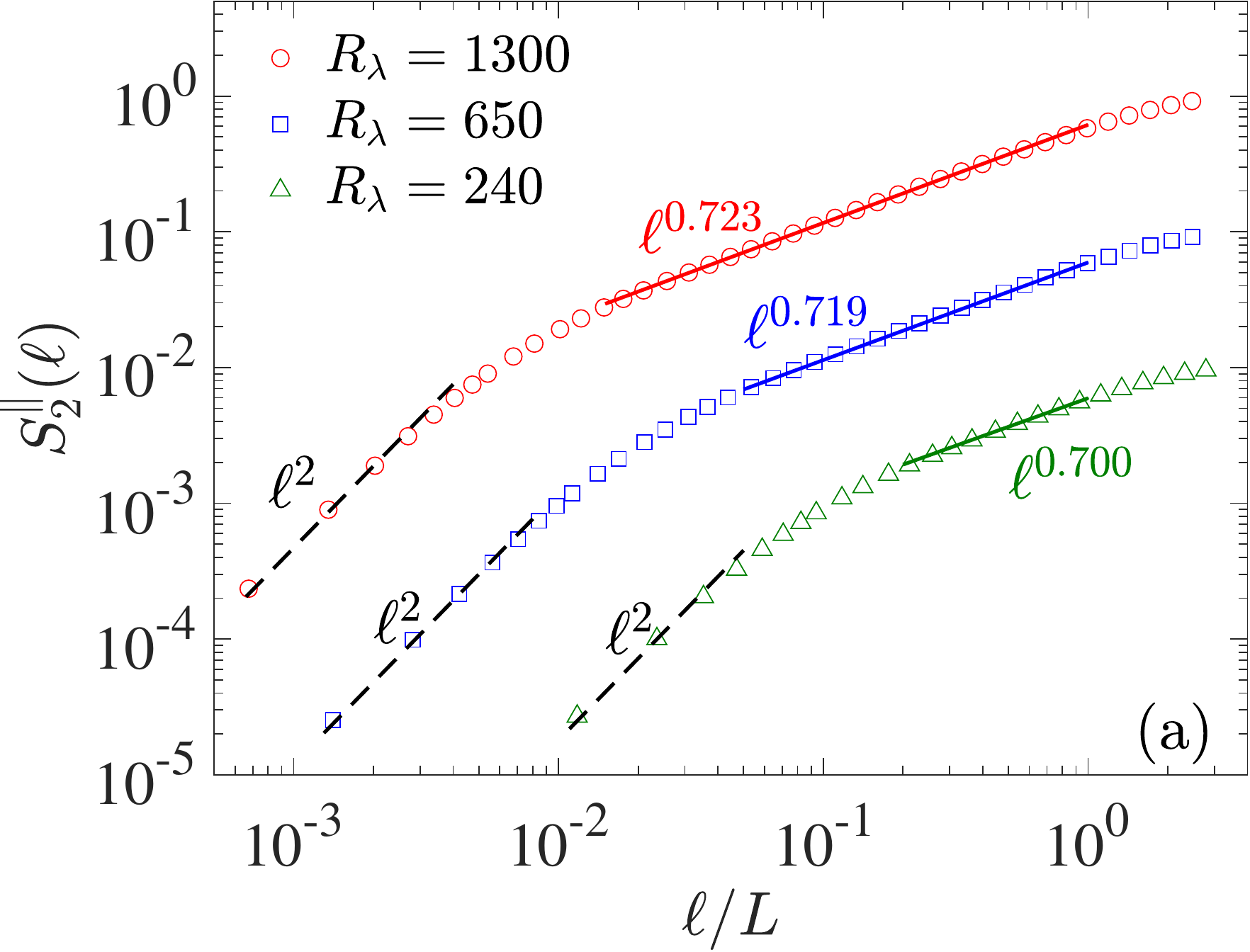}  \hspace{3mm}
\includegraphics[width=0.47\linewidth]{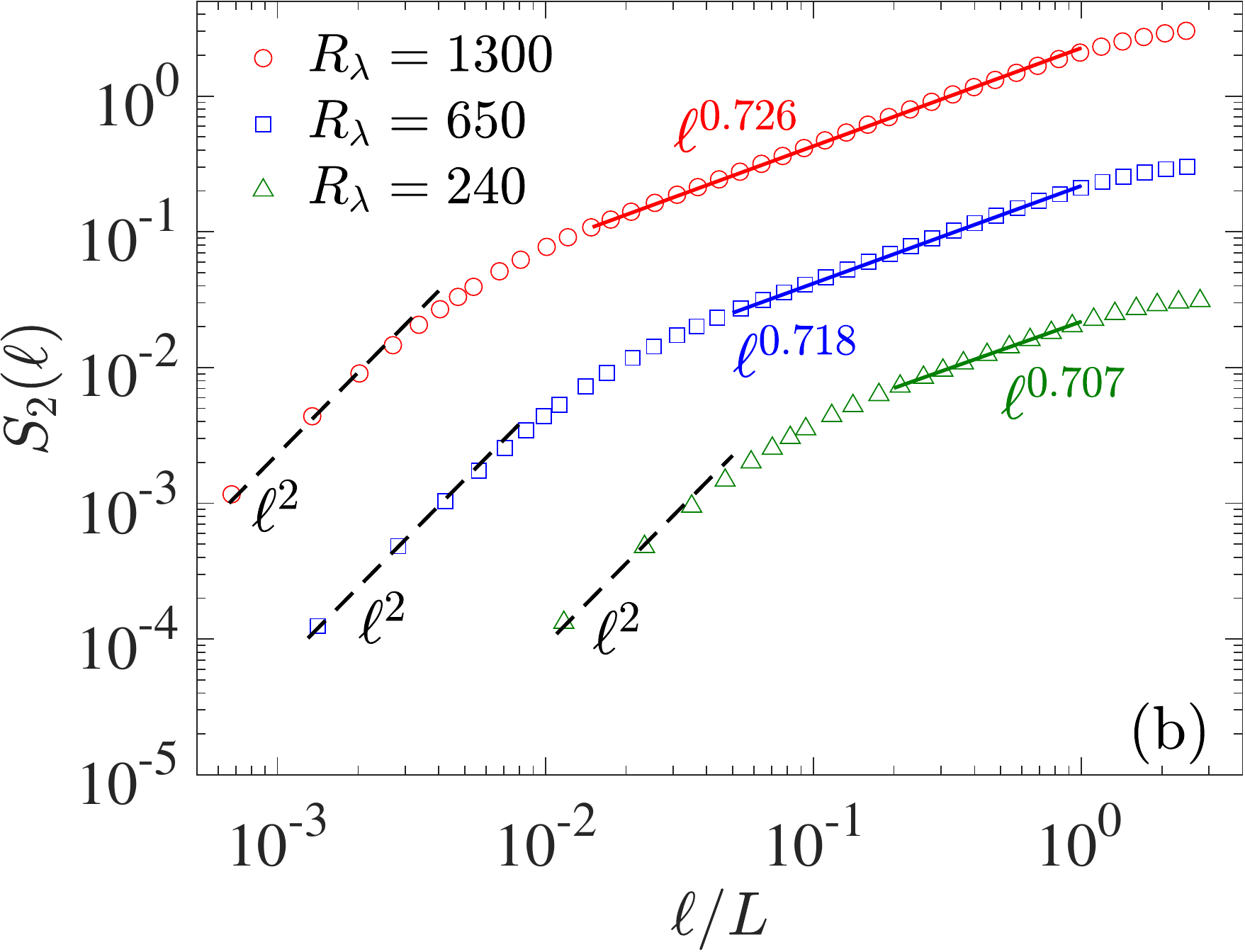}
\caption{Second-order  \emph{longitudinal} (a) and \emph{absolute} (b) structure functions computed from direct numerical simulation of forced homogenous isotropic turbulence with  Taylor scale Reynolds numbers ($R_\lambda:= U \lambda/\nu$ where $U:= \langle |u|^2\rangle^{1/2} $ and $\lambda:= \langle |u|^2\rangle^{1/2}/ \langle |\nabla u|^2\rangle^{1/2}$ ) ranging from $R_\lambda =240 \ (\textcolor{ao(english)}{{\rm green}}),  650 \ (\textcolor{blue}{{\rm blue}}), 1300  \ (\textcolor{red}{{\rm red}})$.  They  exhibit scaling over an inertial range which extends as Reynolds increases.  A best-fit power-law exponent $\zeta_2$ for the power-law $|r|^{\zeta_2}$ in this range is included.  Data from \cite{ISY20},  Fig. 3(a).  }  \label{figure1}
\end{figure}

Onsager took a further step by recognizing that
the behavior \eqref{zerothLaw} requires the fluid to develop singularities as $\nu\to 0$ in a mathematically precise sense.  Specifically, 
for  \eqref{zerothLaw} to occur on sequences of Navier-Stokes solutions, the $p$th order absolute structure functions \emph{cannot} satisfy a bound  of the type 
\be\label{SFbnd}
S_p(\ell) = \int_0^T \int_{\mathbb{T}^d} |u^\nu(x+\ell,t)-u^\nu(x,t)|^p \rmd x \rmd t \leq C |\ell|^{\zeta_p} ,  \qquad \forall |\ell|\leq L,
\ee
for any $\zeta_p>p/3$, $ p\geq 3$ and a constant $C$ independent of viscosity. This assertion, originally stated by Onsager  \cite{O49} about weak solutions of the Euler equation and in the slightly more restrictive setting of H\"{o}lder spaces, has since been rigorously proved \cite{GLE94,CET94,CCFS}. In fact, energy dissipation must vanish as viscosity goes to zero for any family of solutions $\{u^\nu\}_{\nu>0}$ which are uniformly bounded in the Besov space\footnote{A vector  field $v$ belongs to the Besov
space $B_p^{\sigma,\infty}({\mathbb T}^d)$ for $p\geq 1,$ $\sigma\in (0,1)$ at time $t$ if and only if 
\be
\|v(\cdot,t)\|_{L^p}^p<C_0(t), \qquad S_p(\ell,t) \leq C_1(t)\left|\frac{\ell}{L}\right|^{\zeta_p}, \ \forall |\ell|\leq L 
\label{space-struc-fun} \ee 
with $\zeta_p=\sigma p$, $L>0$ and $C_0,C_1\in L^1(0,T)$.  Uniform boundedness of the family $\{u^\nu\}_{\nu>0}$ in $L^{p}(0,T;B_p^{\sigma,\infty}({\mathbb T}^d))$ 
is equivalent to the condition that coefficients $C_0(t),$ $C_1(t)$ independent of $\nu>0$ exist so that the 
bounds (\ref{space-struc-fun}) are satisfied for a.e. $t\in [0,T]$.
 }
 $L^p(0,T;B_p^{1/3+, \infty}(\mathbb{T}^d))$ for $p\geq 3$ \cite{DE19}. 
Thus, Kolmogorov's 1941 theory corresponds to turbulent solutions possessing the maximal degree of smoothness consistent with their ability to anomalously dissipate energy.

It is well known that real fluids do not conform exactly to Kolmogorov's prediction.\footnote{However, weak Euler solutions with less regularity $u\in C^{1/3-}([0,T]\times\mathbb{T}^d)$ and which do not conserve energy have been constructed \cite{I18} after a long series of works \cite{Sch93,Shn97,LS12} and they can be made strictly (globally)  dissipative \cite{BLSV17}. See \cite{BV19} for a very nice recent review of the subject.
 In a sense, these solutions exhibit exact K41-type behavior, although they are not known to arise as physical limits of Leray solutions of Navier-Stokes as required to make contact with real-world high-$Re$ flows.} Intermittency, or  spottiness / non-uniformity of the velocity's roughness and the energy dissipation rate, result in deviations  of the scaling exponents $\zeta_p^\|$ (and $\zeta_p$) from a linear behavior in $p$  \cite{B93,SSJ93,A84,S93,S94,S96}.\footnote{In fact, there is a rigorous connection between these two irregularities: Isett proved \cite{I17} that if $\zeta_p\geq p/3$ for some $p>3$, then the dissipation would have to take place on a full measure set. As experiments indicate that the dissipation takes place on measure zero set of  (spatial) fractal dimension $\approx 2.87$ \cite{MS91}, this is consistent with velocity intermittency $\zeta_p<p/3$ for all $p>3$.  On the other hand, velocity irregularity is not enough to sustain anomalous dissipation: Shvydkoy \cite{S09} proved that ``ordered singularities" with $\zeta_p= 1$ for $p\geq 1$ such as tangential velocity discontinuities across smooth co-dimension one hypersurfaces (regular vortex sheets) conserve energy. }  Experiments do however indicate that for $p$ near three, the formula $\zeta_p\approx p/3$  approximately holds with $\zeta_2\in \frac{2}{3} + [0.03,0.06]$ and $\zeta_3\approx 1$. 
For example, in flow past a sphere $\zeta_2\approx 0.701$ is reported in \cite{B93} and $\zeta_2\approx 0.71$ in \cite{A84} (see Table 2 therein).
 Recent high-resolution numerical simulations report $\zeta_2\approx 0.725$ (see Figure \ref{figure1} and  \ref{figure3}). Although there are slight variations, all these results conform to $\zeta_2 \gtrapprox 2/3$  and $\zeta_3 \lessapprox 1$. These observations motivate:
 \begin{quest}
Why does high-Reynolds number turbulence seems to be as rough as required
to support anomalous dissipation of energy but not much rougher?
 \end{quest}

In this direction, we note that Kolmogovov's prediction \eqref{SFscaling} in the case $p=3$ has a privileged status in that it can be derived (under certain technical assumptions, see Prop. \ref{prop45}) from the equations of motion \eqref{NSE}--\eqref{incom} rather than being merely a consequence of statistical hypotheses.  Specifically, Kolmogorov established the ``4/5--law"  (in dimension three) under only the assumption of anomalous dissipation \eqref{zerothLaw}:
\be\label{45law}
S_3^\|(\ell) :=\langle (\delta_\ell u\cdot \hat{\ell} )^3\rangle \approx -\frac{12}{d(d+2)}\ve \ell,
\ee
which holds in the limit of large Reynolds number $\nu\to 0$ and subsequently small scales $\ell \to 0$.  
 In practice,   \eqref{45law} is observed to hold approximately over the inertial range; see Figure  \ref{figure2} for evidence from \cite{ISY20}. The 4/5--law captures some aspects of the turbulent cascade: energy is transferred through scale by a cubic nonlinear flux term related to $S_3^\|(\ell)$ until it is removed, $\ve>0$, from the system by infinitesimal viscosity.

 \begin{figure} [h!]
 \centering
\includegraphics[width=0.55\linewidth]{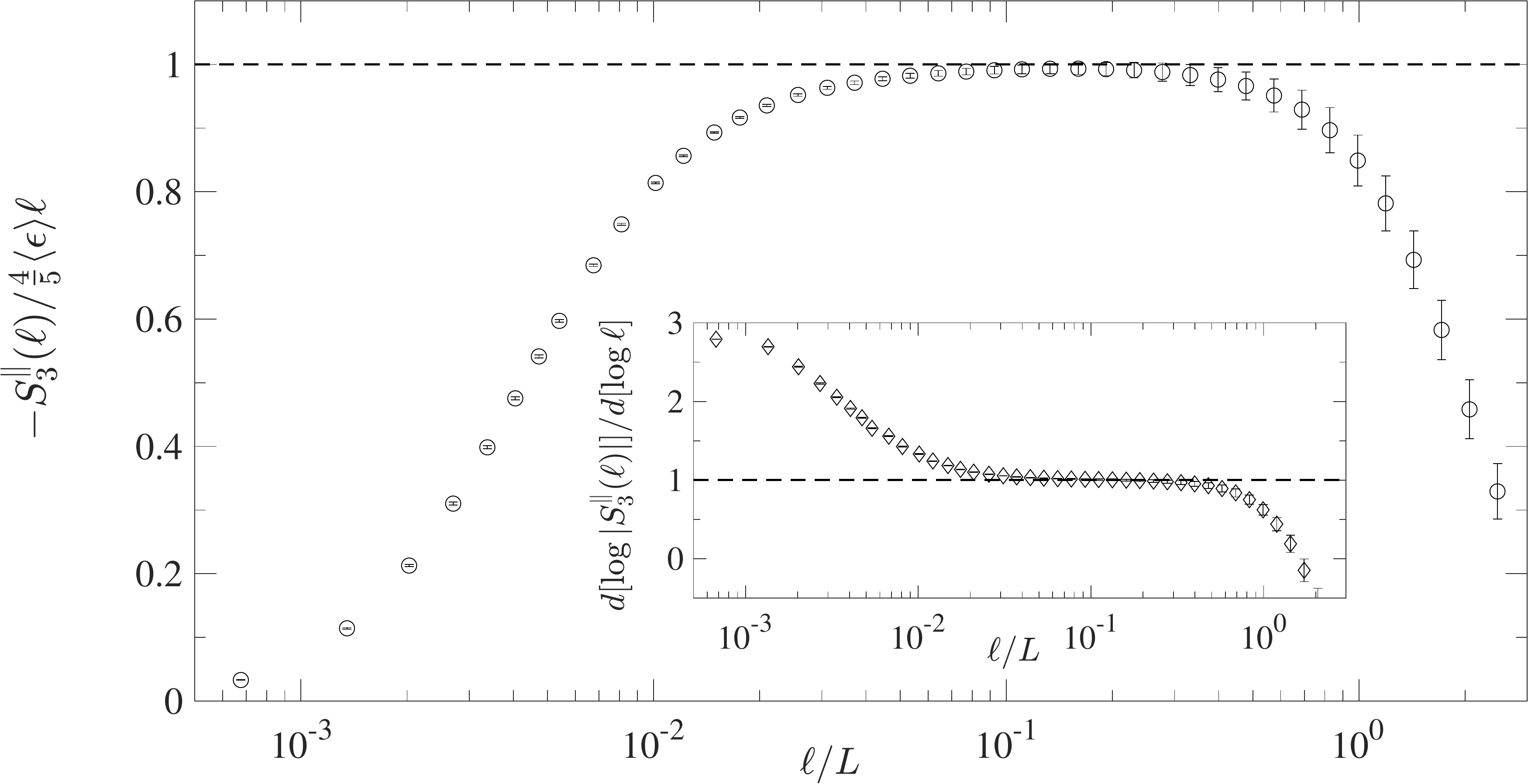} 
\caption{The quantity $-S_3^{\|}(\ell)/(\frac{4}{5}\ell \langle \ve \rangle)$, is plotted for Reynolds numbers $R_\lambda =1300$.  The range of scales over which a value near unity is achieved visibly extends as Reynolds number increases.   Data from \cite{ISY20},  Fig. 1. 
   } \label{figure2}
\end{figure}

In fact, the 4/5--law \eqref{45law} fixes the scaling exponent $p/3$ in \eqref{SFscaling} via the statistical  assumption of monofractal scaling in Kolmogorov's theory.    On its face,  this indicates that the turbulent fluid velocity satisfies $\delta_\ell u \sim \ell^{1/3}$ and so is ``$1/3$ differentiable", at least in some averaged sense.  More precisely, the fact that $\ve$ is apriori bounded by initial data and forcing via equation \eqref{viscousDiss} ensures that  $S_3^\|(\ell) /\ell$ is controlled uniformly for small $\ell$. This suggests that some kind of apriori regularity information --  a `turbulent energy estimate' -- might be extracted from the  4/5--law.
Unfortunately, aside from justifying the assumptions necessary for a rigorous derivation of  \eqref{45law},  there are two obstructions to realizing this hope: (i) the longitudinal structure function does not measure velocity variations in all directions, only those aligned with the separation vector and (ii) the integrand is not sign-definite.  In particular, although a certain skewness is implied by the 4/5--law (positive of $\ve$ means negativity of $\delta_\ell u \cdot \hat{\ell}$ in an averaged sense), it is conceivable that there are large fluctuations which cancel in the integral to yield \eqref{45law} but would disturb this relation if the increments were replaced by their absolute values.  Both of these issues prevent the control on third order longitudinal structure function afforded by   \eqref{45law} from being coercive and so it seems that no direct information about regularity of the velocity can be  immediately  extracted.

Here we  explore  possible  nonlinear mechanisms to extract regularity from 4/5--law.   Our results will be of a conditional nature, involving  hypotheses  which are unproved but which are corroborated by experiment and simulations of  turbulence. Roughly,  our main assumptions  (Hypothesis \ref{hyp45thsLaw} below) are that
 \begin{enumerate}[label=(\alph*)]
\item the Kolmogorov $4/5$--ths law holds
\item the Kolmogorov $4/3$--rds law holds
\end{enumerate}
as well as (Hypothesis \ref{hypothesis} below) 
 \begin{enumerate}[label=(\alph*)]
\item there exists an $\alpha\in [0,1)$ and $C>0$ independent of $\nu$ such that for all scales $\ell>0$
\be\label{bigass*}
- \int_0^T\int_{\mathbb{T}^d}   \left\langle  ( \delta_\ell u^\nu \cdot \hat{\ell})^3\right\rangle_{ang}\rmd x \rmd t \geq C  |\ell|^\alpha \int_0^T\int_{\mathbb{T}^d}   \left\langle  | \delta_\ell u^\nu \cdot \hat{\ell}|^3\right\rangle_{ang}\rmd x \rmd t,  
\ee
\item
there exists an $\beta\in [0,1)$ and $C'>0$  independent of $\nu$ such that for all scales $\ell>0$
\be\label{bigass2*}
- \int_0^T\int_{\mathbb{T}^d}   \left\langle   (\delta_\ell u^\nu \cdot \hat{\ell}) | \delta_\ell u^\nu|^2\right\rangle_{ang} \rmd x \rmd t \geq C'  |\ell|^\beta \int_0^T\int_{\mathbb{T}^d}   \left\langle  | \delta_\ell u^\nu|^3\right\rangle_{ang}\rmd x \rmd t.
\ee
\end{enumerate}
The first hypotheses assert the validity of the Kolmogorov laws for weak solutions.  Although anticipated to be true (the $4/5$--ths law is sometimes referred to as the only ``exact" law in turbulence),  to this day it has not been unconditionally established (see  \cite{E02,B18} for conditional validations and Figure \ref{figure2} for empircal evidence).
The second hypotheses concern some effective alignment properties of the velocity increments with their separation vectors.  A detailed discussion is deferred  to the subsequent section.
As stated above, they are  slightly stronger than Hypothesis \ref{hypothesis} required in the proof, but may be more convenient to verify numerically or in experiment. In fact, there has already been direct evidence of the behavior \eqref{bigass*} with $\alpha\approx 0.03$ from experiment \cite{S93,S94,S96} (see extended discussion in Remark \ref{experiment} below).  In practice, \eqref{bigass*} and \eqref{bigass2*} need only be checked over the finite range of scales in the inertial range
$\ell_\nu \ll \ell \ll L$.
We prove

\begin{thm}\label{theorem}
Let $u$ be a weak solution of the Euler equations of class $L^3(0,T;L^3(\mathbb{T}^d))$. Then
 \begin{enumerate}[label=(\alph*)]
 \item if Hypotheses \ref{hyp45thsLaw}(a) and   \ref{hypothesis}(a) hold, then $u\in L^2(0,T;B_2^{(1-\alpha)/3,\infty}(\mathbb{T}^d))$,
  \item  if Hypotheses \ref{hyp45thsLaw}(b) and   \ref{hypothesis}(b) hold,  then $u\in L^3(0,T;B_3^{(1-\beta)/3,\infty}(\mathbb{T}^d))$. 
\end{enumerate}
\end{thm}
Theorem \ref{theorem}(b) is not very surprising since Hypothesis \ref{hypothesis}(b) (nearly) assumes control on the absolute structure function by the longitudinal.   We include it because there is experimental evidence of such control, because $L^3$ seems to be the natural scale for regularization by the energy cascade. It also applies unconditionally to the entropic solutions of the Burgers equation (see Remark \ref{Burgers}). 
On the other hand, Theorem \ref{theorem}(a) produces regularity in $L^2$ \emph{without} assuming that the full velocity increment can be controlled by the longitudinal component. This is due to the following Lemma which shows that the dynamical law (Euler or Navier-Stokes equations) of the fluid can be used to  deduce information on the full velocity increment from partial information on the behavior of the component in the direction of its separation vector.

\begin{lemma}\label{lemma1}
A weak solution of the incompressible Euler equations is of class  $ L^2(0,T;B_2^{\zeta_2^\|/2,\infty} (\mathbb{T}^d))$ with $\zeta_2^\|\in (0,2]$ if and only if the longitudinal structure function defined by \eqref{structurefunctions} satisfies 
\be\label{s2bound}
\int_0^T \left\langle  S_2^\|(\ell)\right\rangle_{ang} \rmd t \lesssim  |\ell|^{\zeta_2^\|}, \qquad \forall\  |\ell|>0.
\ee
\end{lemma}

\begin{rem}[Weak solutions as zero-viscosity limits]
Lemma \ref{lemma1} has some implications for the weak inviscid limit.  In particular, uniform boundedness of the family $\{u^\nu\}_{\nu>0}$ in $L^2(0,T;B_2^{\zeta_2^\|/2,\infty} (\mathbb{T}^d))$ is equivalent to a bound of the form \eqref{s2bound} independent of viscosity.  In fact, as in Lemma 1 of \cite{DN18},  for Leray solutions $u^\nu \in L^\infty(0,T;L^2(\mathbb{T}^d))\cap L^2(0,T;H^1(\mathbb{T}^d))$ the condition \eqref{s2bound} is equivalent to
\be\label{s2boundfin}
\int_0^T \left\langle  S_2^\|(\ell;\nu)\right\rangle_{ang} \rmd t \lesssim |\ell|^{\zeta_2^\|}, \qquad   \eta(\nu)\leq |\ell|\leq L,
\ee
where $\eta(\nu)= \nu^{1/2(1-s)}$.  Thus, a uniform scaling with any positive exponent of the longitudinal structure function in the ``inertial range" suffices to obtain weak Euler solutions in the inviscid limit (see Thm 1 of \cite{DN18}).   We emphasize that the bound \eqref{s2boundfin} is not naively a compactness statement, although for equations structurally similar to Navier-Stokes, Lemma \ref{s2bound} transforms it into one.
See Figure \ref{figure3} for empirical verification of  Lemma \ref{lemma1} as it applies to inviscid limits of Navier-Stokes turbulence,  relating the bounds (scalings) of the absolute and longitudinal structure functions.
\end{rem}

\begin{figure} [h!] 
\centering
\includegraphics[width=0.50\linewidth]{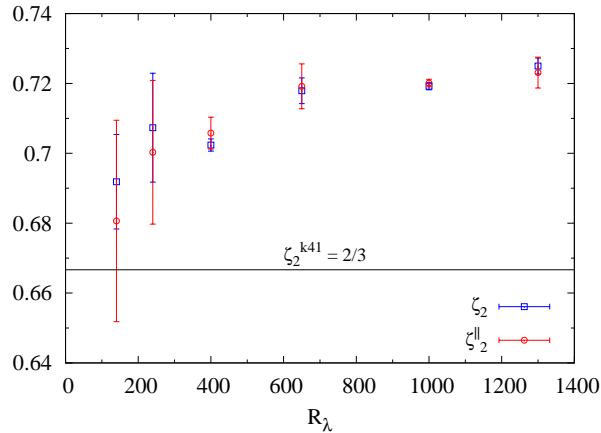} 
\caption{
Best-fit exponents within the inertial range are plotted for absolute $\zeta_2$ and longitudinal $\zeta_2^{\|}$ structure functions. The K41 value of $\zeta_{2}^{k41}:=2/3$ is given for reference. The exponents $\zeta_2, \zeta_2^{\|}$ appear to saturate at $0.725$,  which serves as the exponent which provides the uniform (in $Re$) bounds  \eqref{SFbnd} and \eqref{s2bound}.    Data from \cite{ISY20},  Fig. 3(b).  } \label{figure3}
\end{figure}

 \section{Kolmogorov 4/5--law and Alignment Hypotheses}
 
We here recall a rigorous formulation of the Kolmogorov 4/5--law  for weak solutions of the Euler equations arising as zero viscosity limits and introduce the precise Hypotheses under which our Theorem is established.
 As discussed above, Onsager conjectured \cite{O49} that sufficiently rough, dissipative weak solutions of the Euler equations are candidate descriptions of  high-Reynolds number flows exhibiting the behavior \eqref{zerothLaw}. Onsager's vision of weak Euler solutions as a framework to study zero-viscosity limits follows from sufficient compactness. Indeed, if a family of Leray solutions $\{u^\nu\}_{\nu>0}$ is precompact in $L^3(0,T;L^3(\mathbb{T}^d))$, then strong space-time $L^3$--limits exist $u^\nu\to u\in L^3(0,T;L^3(\mathbb{T}^d))$ and are weak solutions of the Euler equations.\footnote{Such compactness is implied if the family of Leray solutions $\{u^\nu\}_{\nu>0}$ is uniformly--in--$\nu$ bounded in  $L^3(0,T;B_3^{s,\infty}(\mathbb{T}^d))$ \emph{for any} $s>0$ \cite{DE19,DN18}. As discussed above, this is robustly observed in experiments and simulations \cite{B93,SSJ93,S93,S94,S96}.} In fact, this assumption also guarantees the existence of a limiting dissipation measure $\varepsilon[\bu]$,
 \be
\Dlim_{\nu\to 0} \varepsilon^\nu[\bu^\nu] = \varepsilon[\bu] \geq 0,
 \ee
 where the limit is understood in the sense of distributions (it holds upon pairing with any smooth test function and is denoted by $\Dlim$).
 Furthermore, Duchon and Robert \cite{DR00} showed that any weak solution of the Euler equations $u$ of class $L^3(0,T;L^3(\mathbb{T}^d))$ satisfies a (weak) energy balance
\be \partial_t\left(\frac{1}{2}|u|^2\right)+\nabla\cdot\left[\left(\frac{1}{2}|u|^2+p\right)u\right] = -D[u] \label{Ebal} \ee 
where the `inertial dissipation' $D[u]$ is defined by the distributional limit
 \begin{equation}\label{DuchonRobertAnom}
 D[u] : = \Dlim_{\ell \to 0} \frac{1}{4\ell} \int_{\mathbb{T}^d}\! (\nabla \varphi)_\ell(r) \cdot \delta_r\bu(x,t) |\delta_r\bu(x,t)|^2 \rmd r
 \end{equation}
 with $\varphi$ an arbitrary standard mollifier, $(\nabla \varphi)_\ell(r)= \ell^{-d} \nabla \varphi(r/\ell)$ and  $\delta_r\bu(x,t)=u(x+r,t)- u(x,t)$. The distribution defined by \eqref{DuchonRobertAnom} represents the flux of energy into or out of the fluid due to a nonlinear inertial cascade to zero length-scale facilitated, as Onsager envisioned, by sufficiently irregular velocity fields.
 As a consequence of \eqref{Ebal},
  the inertial dissipation matches onto the viscous dissipation anomaly
\be 
D[u] =\varepsilon[u],  \label{flux-anom2}
 \ee
 and the distribution $ D[u]$ must be non-negative and independent of the mollifier $\varphi$.  This independence can be seen directly provided that $u$ has some additional spatial continuity which is made precise by the following

\begin{hyp}\label{hyp45thsLaw}
Let $u$ be any weak solution of the Euler equation of class $L^3(0,T;L^3(\mathbb{T}^d))$.  Suppose 
 \begin{enumerate}[label=(\alph*)]
\item the following version of the Kolmogorov $4/5$--law holds
\begin{align}\label{lim45ths}
\Dlim_{|\ell|\to 0} \frac{1}{|\ell|} \left\langle  ( \delta_\ell u \cdot \hat{\ell})^3\right\rangle_{ang}  &=D_{4/5}^*[u],
\end{align}
\item the following version of the Kolmogorov $4/3$--law holds
\begin{align}\label{lim43rds}
\Dlim_{|\ell|\to 0} \frac{1}{|\ell|}\left\langle   (\delta_\ell u \cdot \hat{\ell}) | \delta_\ell u|^2\right\rangle_{ang} &=D_{4/3}^*[u],
\end{align}
\end{enumerate}
where the angle average denotes $\left\langle f(\ell) \right\rangle_{ang} := \fint_{S^{d-1}} f(\ell) \ \rmd \omega(\hat{\ell})$ and $ \rmd \omega$  is the measure on solid angles.
\end{hyp}

With Hypothesis \ref{hyp45thsLaw} in hand, we see explicitly that $D[u]$ does not depend on the arbitrary mollifier and the standard versions of the Kolmogorov laws hold:

\begin{prop}[\cite{DR00} \& \cite{E02}]\label{prop1}
Under Hypothesis \ref{hyp45thsLaw}, the following distributional equalities hold
 \be\label{45thand43rd}
D_{4/5}^*[u]= -\frac{12}{d(d+2)} D[u], \qquad D_{4/3}^*[u]= -\frac{4}{d} D[u].
\ee
 \end{prop}\label{prop45}
See also discussion in \cite{D19}, which gives a Lagrangian interpretation of these distributions.  In combination with \eqref{flux-anom2} which holds for strong vanishing viscosity limits, Proposition \ref{prop1} constitutes precise versions of the celebrated Kolmogorov 4/5 and 4/3--laws (upon setting $d=3$ in \eqref{45thand43rd}):
  \be\label{viscouslimlaws}
D_{4/5}^*[u]= -\frac{12}{d(d+2)} \varepsilon[u], \qquad D_{4/3}^*[u]= -\frac{4}{d} \varepsilon[u].
  \ee
  Equation \eqref{viscouslimlaws} is a rigorous version of \eqref{45law}.
    It should be noted that the above relationships are \emph{local} in that they hold in the sense of space-time distributions. Moreover, they show that the fluxes $D_{4/5}^*[u]$ and $D_{4/3}^*[u]$ are, in fact, non-positive as distributions, implying a form of skewness of the velocity field.

Now, for any weak Euler solution of class $L^3(0,T;L^3(\mathbb{T}^d))$, the inertial dissipation $D[u]$ must be finite upon averaging in space time.
However, it need not be signed.  One consequence of \eqref{flux-anom2} is that the inertial dissipation inherits an a-priori bound in terms of initial data and forcing, and is non-negative for high-Reynolds number flows exhibiting anomalous dissipation \eqref{zerothLaw}:
 \be\label{fluxbnd}
0 < \int_0^T \int_{\mathbb{T}^d} D[u]\ \rmd x \rmd x <\infty.
 \ee
  The balance \eqref{flux-anom2}, which leads to \eqref{fluxbnd}, is related to the direct energy cascade and can be interpreted as the statement that a nonlinear transfer of energy can be sustained even to infinitesimally small scales, where an infinitesimal viscosity can efficiently remove energy from the system.
Our main thesis is that \eqref{fluxbnd} can provide a partial explanation for the discussed smoothness of the weak Euler solutions, provided that the solutions additionally possess certain structural properties.  
More precisely,  we adopt the assumption that

 \begin{hyp}\label{hypothesis} 
 Let $u$ be a weak solution of Euler satisfying Hypothesis \ref{hyp45thsLaw}. Suppose in addition that
 \begin{enumerate}[label=(\alph*)]
\item there exists an $\alpha\in [0,1)$ and $C:=C(d,T,u_0,f)>0$ such that
\be\label{bigass}
-\int_0^T \int_{\mathbb{T}^d} D_{4/5}^*[u] \rmd x \rmd t\geq  \limsup_{|\ell|\to 0} \frac{C}{ |\ell|^{1-\alpha}} \int_0^T\int_{\mathbb{T}^d}   \left\langle  | \delta_\ell u \cdot \hat{\ell}|^3\right\rangle_{ang}\rmd x \rmd t.
\ee
\item
 there exists an  $\beta\in [0,1)$ and $C':=C'(d,T,u_0,f)>0$  such that
 \be\label{bigass2}
-\int_0^T \int_{\mathbb{T}^d} D_{4/3}^*[u] \rmd x \rmd t \geq  \limsup_{|\ell|\to 0} \frac{C'}{ |\ell|^{1-\beta}} \int_0^T \int_{\mathbb{T}^d}  \left\langle  | \delta_\ell u |^3 \right\rangle_{ang} \rmd x\rmd t.
\ee
\end{enumerate}
\end{hyp}

\begin{rem}
Clearly Hypothesis \ref{hypothesis}(b) is the stronger assumption since $   | \delta_\ell u \cdot \hat{\ell} |^3  \leq   | \delta_\ell u |^3 $.  In turn, we will show that it leads to a stronger form of regularization on the limit Euler solution.  It should be noted that Hypothesis \ref{hypothesis}(a)  is an assumption on the possible cancellations of the angle average rather than a brute force control of a piece of the velocity increment by the full increment as in Hypothesis \ref{hypothesis}(b).  Such statements are assumptions on the average \emph{anti-alignment} of velocity increments with their separation vectors.
Additionally we remark that the choice of $D_{4/3}^*[u]$ in Hypothesis \ref{hypothesis}(b) was not important.  For our purpose (Theorem \ref{theorem}(b)), it is sufficient that it hold for either distribution $D_{4/3}^*[u]$ or $D_{4/5}^*[u]$ -- the important properties are the finiteness and positivity upon space-time averaging.  In any case, if Hypothesis \ref{hyp45thsLaw} holds, with both limits \eqref{lim45ths} and \eqref{lim43rds} existing, then the distributions are interchangeable in the statement of the \eqref{bigass} and \eqref{bigass2}.
 \end{rem}

\begin{rem}[Evidence of Anti--Alignment]\label{experiment}
There is some experimental evidence \cite{S93,S94,S96} for the type of alignment assumed in Hypothesis \ref{hypothesis}(a), in particular that absolute differences do differ in scaling as in \eqref{bigass}, but very slightly.  Specifically,
for absolutely third-order longitudinal structure functions, experiments find \cite{S94} that (a slightly strengthened version of) Hypothesis \ref{hypothesis}(a) holds with $\alpha \approx 0.03$, i.e. $-\langle (\delta_\ell u)^3\rangle \sim \ell^\alpha \langle |\delta_\ell u|^3\rangle$ in the inertial range; see Table 1 therein.  This behavior of $\alpha \ll 1$ has also been observed in a number other experiments
 \cite{B93,SSJ93}. It should be noted that the experimental measurements are inferred from data of the velocity field along a one-dimensional cut and computed the longitudinal structure by appealing to Taylor's hypothesis, ergodicity and statistical isotropy and homogeneity.  
 \end{rem}

 \begin{rem}[Stochastic Setting]
Rigorously establishing alignment properties such as those appearing in Hypothesis \ref{hypothesis} seems to be a very difficult task.  Moreover, even if true generically, it is quite conceivable that it is false pathwise the setting of the deterministic Navier-Stokes solutions due to non-generic events.  Thus, such properties might be easier established in the stochastically forced or random data setting.  For instance, one might be able to prove the existence of statistically stationary, homogenous isotropic martingale solutions.  It is then plausible that  \eqref{bigass} and \eqref{bigass2} hold for such solutions upon ensemble averaging.  See \cite{B18} for some interesting developments concerning the validity of Hypothesis \ref{hyp45thsLaw}.
 \end{rem}

The fact that inertial dissipation and the direct energy cascade can provide a regularization mechanism the weak solutions is well understood in some model problems, such as 1-dimensional conservation laws \cite{GP11,J09,TT07} as well the dyadic (Desnyansky--Novikov) shell model of turbulence  \cite{CZ16}.  For three-dimensional Navier-Stokes, no form of uniform fractional regularity or self-regularization has ever been rigorously established from first principles, however experiments and simulations ubiquitously indicate that solutions do possess some form of these phenomena \cite{B93,SSJ93,S93,S94,S96}. In particular, as discussed above, measurements of multifractal structure function scaling exponents from over the last 60 years indicate that some turbulent solutions of Navier-Stokes enjoy some limited uniform fractional regularity in $L^p$ spaces.
Under the  Hypothesis \ref{hyp45thsLaw} and \ref{hypothesis}, the latter being of a quantitative nature, we capture some of the smoothing effect of the nonlinearity and obtain a self-regularizing property of dissipative weak Euler solutions by Theorem \ref{theorem}. Thus, the 4/5--law together some alignment properties implies regularization for any such weak Euler solution with a finite positive inertial dissipation (in particular, vanishing viscosity limits).
Of course, our Theorem \ref{theorem} is of a conditional nature in that it relies on two major Hypotheses \ref{hyp45thsLaw} and  \ref{hypothesis}, both of which seem very difficult to prove a-priori. However, the validity of these Hypotheses can be checked in Nature through controlled experiment and in direct numerical simulation (DNS) of the Navier-Stokes equations at high Reynolds number.    In Remark \ref{experiment}, we recalled some existing experimental results concerning the validity Hypothesis \ref{hypothesis}(a). 
 We hope that our Hypotheses  \ref{hyp45thsLaw} and  \ref{hypothesis} will be subject to much further testing and scrutiny.


\begin{rem}[Burgers Equation]\label{Burgers}

The two Hypotheses \ref{hyp45thsLaw} and  \ref{hypothesis}  are true for entropy solutions of the 1-dimensional Burgers equation. In particular, the so-called 1/12th law states
\be\label{1/12thlaw}
\lim_{|\ell|\to 0} \frac{1}{12}\int_0^T\int_{\mathbb{T}^d}\frac{1}{|\ell|} \left\langle (\delta_\ell u)^3\right\rangle_{ang} \rmd x \rmd t = -\int_0^T\int_{\mathbb{T}^d}  \varepsilon(x,t) \rmd x\rmd t,
\ee
where the one-dimensional angle average is defined by $
\left\langle (\delta_\ell u)^3\right\rangle_{ang} = \frac{1}{2}\left[ (\delta_\ell u)^3(|\ell|) +  (\delta_\ell u)^3(-|\ell|)\right].$
Equation \ref{1/12thlaw} is the analogue of the 4/5--law in the setting of Burgers and is rigorously established for vanishing viscosity limits.
If there are countably many shocks, the following can be explicitly computed
\begin{align*}
\lim_{|\ell|\to 0} \int_0^T\int_{\mathbb{T}^d}\frac{1}{|\ell|} \left\langle (\delta_\ell u)^3\right\rangle_{ang} \rmd x \rmd t &= \sum_i  (\Delta u_i(t))^3,
\end{align*}
where $\Delta u_i(t)$ is the jump at the $i$th shock. The Lax entropy condition is that $u^->u^+$ at shocks, or $\Delta u_i(t)<0$.  This means that $\Delta u_i(t)= -|\Delta u_i(t)|$ and our Hypothesis \ref{hypothesis} holds with $\alpha=\beta = 0$. 
This is an example of perfect ``anti--alignment".
In accord with Theorem \ref{theorem}(b), we obtain $u\in L^3(0,T;B^{1/3,\infty}_3(\mathbb{T}))$.  In light of the inclusion $(L^\infty \cap BV)(\mathbb{T}^d)\subset B_p^{1/p,\infty}(\mathbb{T}^d)$, this is consistent with the well known BV regularity of entropy solutions of 1D hyperbolic conservation laws \cite{GP11,J09,TT07}. 
\end{rem}

\section{Proofs}

\begin{proof}[Proof of Lemma \ref{lemma1}] Let $s=\zeta_2^\| /2$.
\textbf{(1) $\implies$ (2)}.  This direction is trivial since
\be
\left\langle  S_2^\|(\ell)\right\rangle_{ang} \leq \sup_{|\ell'|\leq |\ell|} \int_{\mathbb{T}^d}  |\delta_{\ell'} u\cdot \hat{\ell'}|^2 \rmd x \leq \sup_{|\ell'|\leq |\ell|} S_2(\ell')\leq C |\ell|^{2s}
\ee
where the Besov regularity $u\in L^2(0,T;B_2^{s,\infty} (\mathbb{T}^d))$ was used in the final inequality.\\

\noindent \textbf{(2) $\implies$ (1)}.   
Let $u\in L^2(0,T;L^2(\mathbb{T}^d))$ be a weak solution of the incompressible Euler equations:
\be\label{weakform}
\int_0^T \int_{\mathbb{T}^d} u \partial_t \varphi \ \rmd x\rmd t + \int_0^T \int_{\mathbb{T}^d} u\otimes u : \nabla \varphi \ \rmd x\rmd t =0.
\ee
where $\varphi := \varphi(x,t) \in C_0^\infty([0,T]\times \mathbb{T}^d)$ is a compactly supported divergence-free test function.  Defining $\varphi_\ell:=\varphi(x-\ell)$. Introducing the increment field $\delta_\ell u := u(x+\ell)-u(x)$ and choosing the test function $\varphi_\ell- \varphi$ in \eqref{weakform} shows that
\be\label{incrementweak}
(\partial_t + u \cdot \nabla )\delta_\ell u= - \nabla_x \delta_\ell p - \delta_\ell u \cdot \nabla_\ell \delta_\ell u +\delta_\ell f
\ee
holds in the sense of distributions.
To obtain this we denote $u'=u(x+\ell)$ and $u=u(x)$ and derive a weak form of the `doubling variables' identity
\begin{align*}
 \int_{\mathbb{T}^d} (u\otimes u : \nabla \varphi_\ell -  u\otimes u : \nabla \varphi) \ \rmd x
&=  \int_{\mathbb{T}^d}  ( \delta_\ell u \otimes  u' - u\otimes   \delta_\ell u) : \nabla \varphi \ \rmd x\\
&=  \int_{\mathbb{T}^d}  ( \delta_{-\ell} u \otimes  u  :\nabla \varphi_\ell - u\otimes   \delta_\ell u: \nabla \varphi)\ \rmd x \\
&=  \int_{\mathbb{T}^d}  ( \delta_{-\ell} u \otimes  u  :\nabla_\ell \varphi_\ell - u\otimes   \delta_\ell u: \nabla \varphi)\ \rmd x \\
&=  \nabla_\ell\cdot \int_{\mathbb{T}^d}   \delta_{\ell} u \otimes u'  \cdot \varphi\ \rmd x -  \int_{\mathbb{T}^d}   u\otimes   \delta_\ell u: \nabla \varphi\ \rmd x \\
&=  \nabla_\ell\cdot \int_{\mathbb{T}^d}   \delta_{\ell} u \otimes \delta_{\ell} u  \cdot \varphi\ \rmd x -  \int_{\mathbb{T}^d}   u\otimes   \delta_\ell u: \nabla \varphi\ \rmd x,
\end{align*}
where we used the fact that $u$ is distributionally divergence-free. This establishes that \eqref{incrementweak} holds in the sense of distributions.  
Dotting the above equation with $\ell$, we find
\be
(\partial_t + u \cdot \nabla)(\delta_\ell u\cdot \ell)=  - \ell \cdot\nabla_x \delta_\ell p - \delta_\ell u \cdot \nabla_\ell  (\delta_\ell u\cdot \ell)+ |\delta_\ell u|^2 + \delta_\ell f.
\ee
Integrating this balance over the torus, we have
\be\label{intbalance}
\int_{\mathbb{T}^d}  |\delta_\ell u|^2 \rmd x=- \nabla_\ell \cdot  \int_{\mathbb{T}^d}\delta_\ell u  (\delta_\ell u\cdot \ell)\ \rmd x,
\ee
where we used only periodicity of the solution fields.
Averaging (in the separation vector $\ell$) equation \eqref{intbalance} over a ball of radius $L$ centered at zero, we find
\be\label{ident1}
\fint_{B_L(0)} \int_{\mathbb{T}^d}  |\delta_\ell u|^2 \rmd x\rmd \ell= \fint_{S^{d-1}}  \left.\left(\int_{\mathbb{T}^d} |\delta_\ell u\cdot\hat{\ell}|^2\ \rmd x\right)\right|_{|\ell|=L}  \rmd \omega(\hat{\ell})
\ee
where $S^{d-1}$ is the unit sphere in $d$-dimensions and $ \rmd \omega$  is the measure on solid angles (unit Haar measure on $S^{d-1}$) and $\fint_{A}:= \frac{1}{|A|} \int_{A}$.
For $p\geq 1$, we define the absolute and longitudinal structure functions to be
\be\label{structurefunctions}
S_p(\ell):=\int_{\mathbb{T}^d}  |\delta_\ell u|^p \rmd x, \qquad S_p^\|(\ell):=\int_{\mathbb{T}^d}  (\delta_\ell u\cdot \hat{\ell})^p \rmd x.
\ee
Introducing the angle-averaging operation
\be
\left\langle f(\ell) \right\rangle_{ang} := \fint_{S^{d-1}} f(\ell) \ \rmd \omega(\hat{\ell}),
\ee
from \eqref{ident1} we have deduced the identity
\be\label{S2ident}
\fint_{B_L(0)}S_2(\ell)\  \rmd \ell = \left\langle  S_2^\|(L\hat{\ell})\right\rangle_{ang}.
\ee

\begin{rem}
In general, the tensor product with $\ell$ yields
\be\nonumber
(\partial_t + u^\nu \cdot \nabla_x)(\delta_\ell u\otimes \ell) +  \delta_\ell u \cdot \nabla_\ell (\delta_\ell u\otimes \ell)  =  \delta_\ell u\otimes \delta_\ell u  - \ell \otimes  \nabla_x \delta_\ell p.
\ee
Thus, integrating over space and separation vectors then yields the tensorial identity
\be\label{zeroident}
 \left\langle \int_{\mathbb{T}^d} (\delta_\ell u\cdot \hat{\ell}) (\delta_\ell u\otimes \hat{\ell})  \rmd x\right \rangle_{ang}=  \int_{B_\ell(0)} \int_{\mathbb{T}^d} (\delta_{\ell'} u\otimes\delta_{\ell'} u)  \rmd x \rmd \ell'.
\ee 
Taking the trace gives the scalar identity \eqref{S2ident}.  We also obtain as a special case in 3D
\be\label{zeroident}
 \left\langle \int_{\mathbb{T}^d} (\delta_\ell u\cdot \hat{\ell}) (\delta_\ell u\times \hat{\ell})  \rmd x\right \rangle_{ang}=0.
\ee 
\end{rem}
Using the identity \eqref{S2ident} together with the assumed bound \eqref{s2bound}, we have
\be
\fint_{B_L(0)}S_2(\ell)\  \rmd \ell \leq C L^{2s}.
\ee
This inequality holds for all $L\geq 0$.  Let $f(\ell):= \| \delta_\ell u\|_{L^2}$ and note that for any $\ell'\in \mathbb{T}^d$, we have
\begin{align*}
|f(\ell)-f(\ell')| &=| \| \delta_\ell u\|_{L^2}-\| \delta_{\ell'} u\|_{L^2}| \leq \|\delta_\ell u - \delta_{\ell'} u\|_{L^2}\\
&= \sqrt{\int_{\mathbb{T}^d} | u(x+\ell) - u(x+ \ell')|^2 \rmd x} = \| \delta_{\ell'-\ell} u\|_{L^2}.
\end{align*}
Thus we have the bound
\be
\fint_{B_L(\ell')} |f(\ell)- f(\ell')|^2\  \rmd \ell \leq \fint_{B_L(\ell')} |f(\ell'-\ell)|^2\  \rmd \ell =  \fint_{B_L(0)} |f(\ell)|^2\  \rmd \ell.
\ee
We conclude that for any $\ell'\in \mathbb{T}^d$ and $L>0$, the following inequality holds
\be
\left( \fint_{B_L(\ell')} |f(\ell)- f(\ell')|\  \rmd \ell\right)^2\leq \fint_{B_L(\ell')} |f(\ell)- f(\ell')|^2\  \rmd \ell\leq   C L^{2s}
 \ee
which follows by Jensen's inequality.
We finally appeal to the following basic fact
\begin{lemma}\label{lemHolder}
Assume that there exists $C>0$ and $\alpha \in (0,1]$ such that for every $x_0\in \mathbb{T}^d$ and $r>0$,
\be
\frac{1}{|B_r(x_0)|} \int_{B_r(x_0)} |f(x)- f(x_0)| \  \rmd x < C r^\alpha.
\ee
Then $f$ is H\"{o}lder continuous with exponent $\alpha$.
\end{lemma}
\begin{proof}
Let $x_0, y_0\in \T^d$ and set $r=|x_0-y_0|$. Then we have $B(y_0, r)\subset B(x_0, 2r)$. Thus,
\[
\begin{aligned}
2Cr^\alpha &\ge \frac{1}{|B_{2r}(x_0)|} \int_{B_{2r}(x_0)} |f(x)- f(x_0)|dx+\frac{1}{|B_{2r}(y_0)|} \int_{B_{2r}(y_0)} |f(x)- f(y_0)|dx\\
&= (c_0(2r)^d)^{-1} \left\{\int_{B_{2r}(x_0)} |f(x)- f(x_0)|dx+ \int_{B_{2r}(y_0)} |f(x)- f(y_0)|dx\right\}\\
&\ge (c_0 (2r)^d)^{-1}\int_{B_r(x_0)} |f(x)- f(x_0)|+|f(x)-f(y_0)|dx\\
&\ge (c_0 (2r)^d)^{-1}\int_{B_r(x_0)} |f(x_0)-f(y_0)|dx=2^{-d}|f(x_0)-f(y_0)|
\end{aligned}
\]
where $c_0$ denotes the volume of the unit ball.  It follows that 
\[
|f(x_0)-f(y_0)|\le 2^{d+1}C|x_0-y_0|^\alpha
\]
for any $x_0, y_0\in \T^d$. This  completes the proof. 
\end{proof}
\noindent Lemma \ref{lemHolder} allows us to conclude that $S_2(\ell)$ is H\"{o}lder continuous in $\ell$ with exponent $2s$ and
\be\nonumber
\|u\|_{L^2(0,T;L^2(\mathbb{T}^d))}<C_0, \ \ \  \| \delta_\ell u\|_{L^2(0,T;L^2(\mathbb{T}^d))} \leq C_1 |\ell|^s \ \ \implies \ \ \|u\|_{L^2(0,T; B_2^{s,\infty}(\mathbb{T}^d))} \leq C_2.
\ee
\end{proof}

\begin{proof}[Proof of Theorem \ref{theorem}(a)] By Jensen's inequality, we have
\begin{align}
\left(\frac{1}{|\ell|^{2(1-\alpha)/3}}\int_0^T\left\langle  S_2^\|(\ell)\right\rangle_{ang}  \rmd t\right)^{3/2} &\leq \frac{1}{|\ell|^{1-\alpha}} \int_0^T \left\langle\int_{\mathbb{T}^d}  | \delta_{\ell} u \cdot \hat{\ell}|^3\rmd x\right\rangle_{ang} \rmd t.
\end{align}
Now, for any $\epsilon>0$, we can choose $\ell_\epsilon$ sufficiently small such that for all $\ell\leq \ell_\epsilon$, we have
\begin{align}\nonumber
\left(\frac{1}{|\ell|^{2(1-\alpha)/3}}\int_0^T\left\langle  S_2^\|(\ell)\right\rangle_{ang}  \rmd t\right)^{3/2} &  \leq 
 \sup_{\ell\leq \ell_\epsilon} \frac{1}{ |\ell|^{1-\alpha}} \int_0^T\int_{\mathbb{T}^d}   \left\langle  | \delta_\ell u \cdot \hat{\ell}|^3\right\rangle_{ang}\rmd x \rmd t \\ \nonumber
&\leq \limsup_{|\ell|\to 0} \frac{1}{ |\ell|^{1-\alpha}} \int_0^T\int_{\mathbb{T}^d}   \left\langle  | \delta_\ell u \cdot \hat{\ell}|^3\right\rangle_{ang}\rmd x \rmd t +\epsilon \\
& \leq -\int_0^T \int_{\mathbb{T}^d} D_{4/5}^*[u] \rmd x \rmd t + \epsilon,
\end{align}
which results from Hypothesis \ref{hypothesis}(a).   It follows from Hypothesis \ref{hyp45thsLaw} and Proposition \ref{prop1} that
\begin{align}\nonumber
\left(\frac{1}{|\ell|^{2(1-\alpha)/3}}\int_0^T\left\langle  S_2^\|(\ell)\right\rangle_{ang}  \rmd t\right)^{3/2}
& \leq \frac{12}{d(d+2)}\int_0^T \int_{\mathbb{T}^d} D[u] \rmd x \rmd t + \epsilon.
\end{align}
 Since $f(x)= x^{2/3}$, $x>0$ is monotone increasing, we have  for all $  \ell$ sufficiently small that
\be
\int_0^T\left\langle  S_2^\|(\ell)\right\rangle_{ang}  \rmd t \leq C_0 |\ell|^{2(1-\alpha)/3},
\ee
where  $C_0:=C_0(d,T, u_0, f)$ depends on magnitude of the inertial energy dissipation (correspondingly the anomalous viscous dissipation if the Euler solution is obtained as a vanishing viscosity limit).  The claimed regularity follows by applying Lemma \ref{prop1}.
\end{proof}

\begin{proof}[Proof of Theorem \ref{theorem}(b)] 
We employ the following definition of Besov spaces.  Fix $\ell>0$ and consider the ``truncated ball" $B_T:=\{n:1/2<|n|<1\}$.  The truncated-ball mean is then
 \be
 (\cV_\ell f)(x):= \frac{1}{|B_T|} \int_{B_T} \!\! \delta f(\ell n;x) \ \rmd n.
 \ee
The norm for the Besov space $B_p^{s,\infty}(\mathbb{T}^d)$ can then be defined as
\be
  \|f\|_{B_p^{s,\infty}} := \|f\|_p + |f|_{B_p^{s,\infty}}, \qquad  |f|_{B_p^{s,\infty}} :=  \sup_{N\geq 0} 2^{sN} \| \cV_{2^{-N}} f\|_p.
\ee
See Appendix C of \cite{E96} and Section 2.5.11--12 of \cite{HT83}.   Our aim is to obtain a non-trivial  $L^3$--integrable upper bound for the Besov norm $ \|u(t)\|_{B_3^{s,\infty}}$ under for some $s>0$.   By assumption  $u\in L^3(0,T;L^3({\mathbb T}^d))$, so we need only to find an integrable upper bound for the Besov semi-norm $|u(t)|_{B_3^{s,\infty}}^3$.
 First, by Jensen's inequality, the $p$th-power of ball-averages is bounded by the ball-average of the $p$th power:
\be\label{jensen}
\| \cV_{\ell} f\|_p^p \leq  \frac{1}{|B_T|} \int_{B_T}  \| \delta f(\ell n;\cdot)\|_p^p \ \rmd n, \qquad p\geq 1.
\ee
Thus,  letting $\ell_N:= 2^{-N}$, we have
\bea 
 |u(t)|_{B_3^{s,\infty}}^3&=& \  \Big(\sup_{N\geq 0} \ell_N^{-s} \| \cV_{2^{-N}} f\|_3\Big)^3 \nonumber
 \\ &\leq& \  \sup_{N\geq 0} \ell_N^{-3s}\frac{1}{|B_T|} \int_{B_T}  \| \delta u(\ell_N n;\cdot,t)\|_3^3 \  \rmd n\nonumber\\
 &=&   \   \sup_{N\geq 0} \ell_N^{-3s} \frac{\ell_N^{-d}}{|B_T|}   \int_{\ell_{N-1}}^{\ell_N} \rmd \rho\ \rho^{d-1} \int_{S^{d-1}}   \|\delta u(\rho \hat{r};\cdot,t)\|_3^3 \ \rmd\omega(\hat{\ell})
 \eea
where $\rho=|r|$ and where we have used the upper bound \eqref{jensen}.  
Fix any scale $\ell_0$ and split into small and large
 \bea\nonumber
 |u(t)|_{B_3^{s,\infty}}^3
  &\leq &      \sup_{|\ell| \leq  \ell_0} \ell^{-3s } \int_{\mathbb{T}^d}   \left\langle  | \delta_\ell u(x,t) |^3\right\rangle_{ang}\rmd x +     \sup_{|\ell| >  \ell_0} \ell^{-3s } \int_{\mathbb{T}^d}   \left\langle  | \delta_\ell u(x,t) |^3\right\rangle_{ang}\rmd x \\
    &\leq &      \sup_{|\ell| \leq  \ell_0} \ell^{-3s } \int_{\mathbb{T}^d}   \left\langle  | \delta_\ell u(x,t) |^3\right\rangle_{ang}\rmd x +     2\ell_0^{-3s } \| u\|_{L^3}^3.
\eea
For $\beta\in (0,1)$, we have
  \bea \nonumber
 \int_0^T |u(t)|_{B_3^{s,\infty}}^3\rmd t
  \leq       \sup_{|\ell| \leq  \ell_0} \ell^{(1-\beta)-3s } \left(\frac{1}{|\ell|^{1-\beta}}  \int_0^T\int_{\mathbb{T}^d}   \left\langle  | \delta_\ell u(x,t) |^3\right\rangle_{ang}\rmd x \rmd t\right)+   2\ell_0^{-3s } \| u\|_{L^3(0,T;L^3(\mathbb{T}^d))}^3.
\eea
It follows that for any $s\leq (1-\beta)/3$, we have for any $\epsilon>0$ that there exists an $\ell_\epsilon$ such that for all $\ell_0\leq \ell_\epsilon$, 
  \bea \nonumber
 \int_0^T |u(t)|_{B_3^{s,\infty}}^3\rmd t
 &\leq&        \ell_0^{(1-\beta)-3s }   \sup_{|\ell| \leq  \ell_0}  \frac{1}{|\ell|^{1-\beta}}  \int_0^T\int_{\mathbb{T}^d}   \left\langle  | \delta_\ell u(x,t) |^3\right\rangle_{ang}\rmd x \rmd t +     2\ell_0^{-3s } \| u\|_{L^3(0,T;L^3(\mathbb{T}^d))}^3 \\\nonumber
  &\leq&        \ell_0^{-3s }   \left( \limsup_{|\ell| \to 0} \frac{1}{|\ell|^{1-\beta}}  \int_0^T\int_{\mathbb{T}^d}   \left\langle  | \delta_\ell u(x,t) |^3\right\rangle_{ang}\rmd x \rmd t+ \epsilon +   2 \| u\|_{L^3(0,T;L^3(\mathbb{T}^d))}^3\right)\\ \nonumber
  &\leq&       \ell_0^{-3s } \left(-\int_0^T \int_{\mathbb{T}^d} D_{4/3}^*[u] \rmd x \rmd t +\epsilon + 2 \| u\|_{L^3(0,T;L^3(\mathbb{T}^d))}^3 \right)\\
    &=&     \ell_0^{-3s }\left( \frac{4}{d}   \int_0^T \int_{\mathbb{T}^d} D[u] \rmd x \rmd t +\epsilon + 2 \| u\|_{L^3(0,T;L^3(\mathbb{T}^d))}^3 \right)
\eea
where we used the fact that the distributional limit exists by Hypothesis \ref{hyp45thsLaw} and employed our main Hypothesis \ref{hypothesis}(b) in passing to the second to last line. We remark that the bound on the Besov semi-norm depends on magnitude of the inertial (or anomalous) dissipation.
 \end{proof}

 \subsection*{Acknowledgments} I am enormously grateful to K. Iyer for providing the Figures \ref{figure1}, \ref{figure2} and \ref{figure3}, as well as for enlightening discussions. I would also like to thank P. K. Yeung for the DNS data, which used the Extreme Science and Engineering Discovery Environment (XSEDE) resource Stampede2 at the Texas Advanced Computing Center through allocation PHY200084.
  I am grateful to G. L. Eyink for many useful conversations; it was he who suggested that there should be a connection between the 4/5--law and self-regularization. I thank also P. Constantin, H. Q. Nguyen and V. Vicol for insightful discussion.  This research was partially supported by NSF-DMS grant 2106233.

\end{document}